\documentclass[letterpaper, 10 pt, conference]{ieeeconf}  

\IEEEoverridecommandlockouts                              
\usepackage{amsmath}
\usepackage{amssymb}
\usepackage{color}
\usepackage[usenames,dvipsnames]{xcolor}
\usepackage{graphicx}
\usepackage{units}

\usepackage{caption}
\usepackage{subcaption}

\usepackage{algorithm}
\usepackage{algorithmic}

\newcommand{\mdp}{\mathcal{M}}
\newcommand{\mec}{\mathcal{N}}
\newcommand{\product}{\mathcal{P}}
\newcommand{\ra}{\mathcal{A}}
\newcommand{\act}{A}
\newcommand{\prob}{\mathbf{P}}
\newcommand{\ap}{\mathrm{AP}}

\newcommand{\run}{\mathrm{Run}}
\newcommand{\runfin}{\mathrm{Run}_\mathrm{fin}}
\newcommand{\probm}{\mathrm{Pr}}

\newcommand{\G}{\mathbf{G}}
\newcommand{\F}{\mathbf{F}}
\newcommand{\X}{\mathbf{X}}
\newcommand{\U}{\mathbf{U}}

\newcommand{\strat}{C}
\newcommand{\pstrat}{\zeta}

\newcommand{\pisur}{\pi_{\mathrm{sur}}}
\newcommand{\sur}{\mathrm{sur}}
\newcommand{\cycles}{\sharp}
\newcommand{\EC}{\mathrm{EC}}
\newcommand{\MEC}{\mathrm{MEC}}
\newcommand{\AEC}{\mathrm{AEC}}
\newcommand{\MAEC}{\mathrm{MAEC}}
\newcommand{\maec}{\mathsf{maec}^*}

\newcommand{\ie}{{\it i.e., }}
\newcommand{\eg}{{\it e.g., }}

\newtheorem{definition}{Definition}
\newtheorem{problem}{Problem}
\newtheorem{theorem}{Theorem}
\newtheorem{proposition}{Proposition}

\title{\LARGE \bf
Optimal Control of MDPs with Temporal Logic Constraints
}

\author{M\'{a}ria Svore\v{n}ov\'{a}, Ivana \v{C}ern\'{a} and Calin Belta
\thanks{M. Svore\v{n}ov\'{a}, I. \v{C}ern\'{a} are with Faculty of Informatics, Masaryk University, Brno, Czech Republic, {\tt\footnotesize svorenova@mail.muni.cz, cerna@muni.cz}. C. Belta is with Department of Mechanical Engineering and the Division of Systems Engineering, Boston University, Boston, MA, USA, {\tt\footnotesize cbelta@bu.edu}. This work was partially supported at Masaryk University by grants GAP202/11/0312, LH11065, and at Boston University by ONR grants MURI N00014-09-1051, MURI N00014-10-10952 and by NSF grant CNS-1035588.}%
}

\begin{document}

\maketitle
\thispagestyle{empty}
\pagestyle{empty}

\begin{abstract}

In this paper, we focus on formal synthesis of control policies for finite Markov decision processes with non-negative real-valued costs. We develop an algorithm to automatically generate a policy that guarantees the satisfaction of a correctness specification expressed as a formula of Linear Temporal Logic, while at the same time minimizing the expected average cost between two consecutive satisfactions of a desired property. The existing solutions to this problem are sub-optimal. By leveraging ideas from automata-based model checking and game theory, we provide an optimal solution. We demonstrate the approach on an illustrative example.


\end{abstract}

\section{INTRODUCTION}

Markov Decision Processes (MDP) are probabilistic models widely used in various areas, such as economics, biology, and engineering. In robotics, they have been successfully used to model the motion of systems with actuation and sensing uncertainty, such as ground robots \cite{LaAnBe-TRO-2011}, unmanned aircraft \cite{uavMDP}, and surgical steering needles~\cite{alterovitz2007stochastic}. MDPs are central to control theory~\cite{bertsekasVolumeII}, probabilistic model checking and synthesis in formal methods~\cite{baierBook,yannakakis95}, and game theory~\cite{competitivemdps}. 

MDP control is a well studied area (see \eg~\cite{bertsekasVolumeII}). The goal is usually to optimize 
the expected value of a cost over a finite time (\eg stochastic shortest path problem) or an average expected cost in infinite time (\eg average cost per stage problem). 
Recently, there has been increasing interest in developing MDP control strategies from rich specifications given as formulas of probabilistic temporal logics, such as 
Probabilistic Computation Tree Logic (PCTL) and Probabilistic Linear Temporal Logic (PLTL)~\cite{dennismdpltlmaxprob,LaAnBe-TRO-2011}. It is important to note that both optimal control and temporal logic control problems for MDPs have their counterpart in automata game theory. Specifically, optimal control translates to solving $1 \nicefrac{1}{2}$-player games with payoff functions, such as discounted-payoff and mean-payoff games~\cite{mdpmeanpayoffenergyparity}. 
Temporal logic control for MDPs  corresponds to solving $1 \nicefrac{1}{2}$-player games with parity objectives \cite{gamesforcs}.

Our aim is to optimize the behavior of a system subject to correctness (temporal logic) constraints. Such a connection between optimal and temporal logic control is an intriguing problem with potentially high impact in several applications. Consider, for example, a mobile robot involved in a persistent surveillance mission in a dangerous area under tight fuel or time constraints. The correctness requirement is expressed as a temporal logic specification, \eg ``Keep visiting A and then B and always avoid C''. The resource constraints translate to minimizing a cost function over the feasible trajectories of the robot. Motivated by such applications, in this paper we focus on correctness specifications given as LTL formulae and optimization objectives expressed as average expected cumulative costs per surveillance cycle (ACPC). 

The main contribution of this work is to provide a sound and complete solution to the above problem.  This paper can be seen as an extension of~\cite{majajanaCDC12,majaACC13,dennisCDC11,janayushanICRA12}. In ~\cite{majajanaCDC12}, we focused on deterministic transition systems and 
developed a finite-horizon online planner to provably satisfy an LTL constraint while optimizing the behavior of the system between every two consecutive satisfactions of a given proposition. 
We extended this framework in~\cite{majaACC13}, where we provided an algorithm to optimize the  long-term average behavior of deterministic transition systems with time-varying events of known statistics. The closest to this work is~\cite{dennisCDC11}, where the authors focus on a problem of optimal LTL control of MDPs with real-valued costs on actions. The correctness specification is assumed to include a persistent surveillance task and the goal is to minimize the long-term expected average cost between successive visits of the locations under surveillance. Using dynamic programming techniques, the authors design a solution that is sub-optimal in the general case. In~\cite{janayushanICRA12}, it is shown that, for a certain fragment of LTL, the solution becomes optimal. By using recent results from game theory~\cite{chatterjeeMFCS11}, in this paper we provide an optimal solution for full LTL.

The rest of the paper is organized as follows. In Sec.~\ref{sec:prelims} we introduce the notation and provide necessary definitions. The problem is formulated in Sec.~\ref{sec:pf}. The main algorithms together with discussions on their complexity are presented in Sec.~\ref{sec:solution}. Finally, Sec.~\ref{sec:casestudy} contains experimental results. 

\section{Preliminaries}\label{sec:prelims}

For a set $\mathsf{S}$, we use $\mathsf{S}^{\omega}$ and $\mathsf{S}^+$ to denote the set of all infinite and all non-empty finite sequences of elements of $\mathsf{S}$, respectively. 
For a finite sequence $\tau = a_0\dots a_n\in \mathsf{S}^+$, we use $|\tau|=n+1$ to denote the length of $\tau$. For $0\leq i\leq n$, $\tau(i)=a_i$ 
and $\tau^{(i)}=\mathsf{a_0\ldots a_i}$ is the finite prefix of $\tau$ of length $i+1$. We use the same notation 
for an infinite sequence from the set $\mathsf{S}^{\omega}$.

\subsection{MDP Control}\label{prelim:mdp}

\begin{definition}\label{def:mdp}
A \emph{Markov decision process} (MDP) is a tuple $\mdp=(S,\act,\prob,\ap,L,g)$, where $S$ is a non-empty finite set of states, $\act$ is a non-empty finite set of actions, $\prob \colon S\times \act \times S \to [0,1]$ is a transition probability function such that for every state $s\in S$ and action $\alpha \in \act$ it holds that $\sum_{s'\in S} \prob(s,\alpha,s')\in \{0,1\}$, 
$\ap$ is a finite set of atomic propositions, $L\colon S \to 2^{\ap}$ is a labeling function, and $g\colon S\times \act \to \mathbb{R}^+_0$ is a cost function.
An \emph{initialized} Markov decision process is an MDP $\mdp=(S,\act,\prob,\ap,L,g)$ with a distinctive initial state $s_{init}\in S$.
\end{definition}

An action $\alpha \in \act$ is called enabled in a state $s\in S$ if 
$\sum_{s'\in S} \prob(s,\alpha,s')=1$. With a slight abuse of notation, $\act(s)$ denotes the set of all actions enabled in a state $s$. We assume $\act(s)\neq \emptyset$ for every $s\in S$.


A run of an MDP $\mdp$ is an infinite sequence of states $\rho = s_{0} s_{1}\ldots \in S^{\omega}$ such that for every $i\geq 0$, there exists $\alpha_i \in \act(s_i)$, $\prob(s_{i},\alpha_{i},s_{i+1})>0$. 
We use $\run^{\mdp}(s)$ to denote the set of all runs of $\mdp$ that start in a state $s\in S$. Let $\run^{\mdp}=\bigcup_{s\in S}\run^{\mdp}(s)$. A~finite run $\sigma=s_0\ldots s_n\in S^+$ of $\mdp$ is a finite prefix of a run in $\mdp$ and $\runfin^{\mdp}(s)$ denotes the set of all finite runs of $\mdp$ starting in a state $s\in S$. Let $\runfin^{\mdp}=\bigcup_{s\in S}\runfin^{\mdp}(s)$. The length $|\sigma|=n+1$ of a finite run $\sigma=s_0\ldots s_n$ is also referred to as the number of stages of the run. The last state of $\sigma$ is denoted by $last(\sigma)=s_n$.

The word induced by a run $\rho = s_0s_1\ldots$ of $\mdp$ is an infinite sequence $L(s_0)L(s_1)\ldots\in (2^{\ap})^{\omega}$. Similarly, a finite run of $\mdp$ induces a finite word from the set $(2^{\ap})^+$.

\begin{definition}
Let $\mdp=(S,\act,\prob,\ap,L,g)$ be an MDP. An \emph{end component} (EC) of the MDP $\mdp$ is 
an MDP $\mec=(S_{\mec},\act_{\mec},\prob|_{\mec},\ap,L|_{\mec},g|_{\mec})$ such that $\emptyset \neq S_{\mec}\subseteq S$, $\emptyset \neq \act_{\mec}\subseteq \act$. 
For every $s\in S_{\mec}$ and $\alpha \in \act_{\mec}(s)$ it holds that $\{s'\in S\mid \prob(s,\alpha,s')>0\}\subseteq S_{\mec}$. For every pair of states $s,s'\in S_{\mec}$, there exists a finite run $\sigma \in \runfin^{\mec}(s)$ such that $last(\sigma)=s'$. We use $\prob|_{\mec}$ to denote the function $\prob$ restricted to the sets $S_{\mec}$ and $\act_{\mec}$. Similarly, we use $L|_{\mec}$ and $g|_{\mec}$ with the obvious meaning. If the context is clear, we only use $\prob,L,g$ instead of $\prob|_{\mec},L|_{\mec},g|_{\mec}$. EC $\mec$ of $\mdp$ is called \emph{maximal} (MEC) if there is no EC $\mec'=(S_{\mec'},\act_{\mec'},\prob,\ap,L,g)$ of $\mdp$ such that $\mec'\neq \mec$, $S_{\mec}\subseteq S_{\mec'}$ and $\act_{\mec}(s)\subseteq \act_{\mec'}(s)$ for every $s\in S_{\mec}$. 
The set of all end components and maximal end components of $\mdp$ are denoted by $\EC(\mdp)$ and $\MEC(\mdp)$, respectively.
\end{definition}

The number of ECs of an MDP $\mdp$ can be up to exponential in the number of states of $\mdp$ and they can intersect. On the other hand, MECs are pairwise disjoint and every EC is contained in a single MEC. Hence, the number of MECs of $\mdp$ is bounded by the number of states of $\mdp$.


\begin{definition}\label{def:strat}
Let $\mdp=(S,\act,\prob,\ap,L,g)$ be an MDP. A~\emph{control strategy} for $\mdp$ is a function $\strat \colon \runfin^{\mdp}\to \act$ such that for every $\sigma\in \runfin^{\mdp}$ it holds that $\strat(\sigma)\in \act(last(\sigma))$.
\end{definition}

A strategy $\strat$ for which $\strat(\sigma)=\strat(\sigma')$ for all finite runs $\sigma,\sigma'\in \runfin^{\mdp}$ with $last(\sigma)=last(\sigma')$ is called memoryless. In that case, we consider $\strat$ to be a function $\strat \colon S\to \act$.
A strategy is called finite-memory if it is defined as a tuple $\strat = (M,\mathsf{act},\Delta,\mathsf{start})$, where $M$ is a finite set of modes, $\Delta\colon M\times S\to M$ is a transition function, $\mathsf{act}\colon M\times S\to \act$ selects an action to be applied in $\mdp$, and $\mathsf{start}\colon S\to M$ selects the starting mode for every $s\in S$.

A run $\rho=s_0s_1\ldots \in \run^{\mdp}$ of an MDP $\mdp$ is called a run under a strategy $\strat$ for $\mdp$ if for every $i\geq 0$, it holds that $\prob (s_i,\strat(\rho^{(i)}),s_{i+1})>0$. A finite run under $\strat$ is a finite prefix of a run under $\strat$. The set of all infinite and finite runs of $\mdp$ under $\strat$ starting in a state $s\in S$ are denoted by $\run^{\mdp,\strat}(s)$ and $\runfin^{\mdp,\strat}(s)$, respectively. Let $\run^{\mdp,\strat}=\bigcup_{s\in S}\run^{\mdp,\strat}(s)$ and $\runfin^{\mdp,\strat}=\bigcup_{s\in S}\runfin^{\mdp,\strat}(s)$.

Let $\mdp$ be an MDP, $s$ a state of $\mdp$, and $\strat$ a strategy for $\mdp$. The following probability measure is used to argue about the possible outcomes of applying $\strat$ in $\mdp$ starting from $s$. 
Let $\sigma \in \runfin^{\mdp,\strat}(s)$ be a finite run. A cylinder set $\mathrm{Cyl}(\sigma)$ of $\sigma$ is the set of all runs of $\mdp$ under $\strat$ that have $\sigma$ as a finite prefix. 
There exists a unique probability measure $\probm^{\mdp,\strat}_{s}$ on the $\sigma$-algebra generated by the set of cylinder sets of all runs in $\runfin^{\mdp,\strat}(s)$. For $\sigma=s_0\ldots s_n\in \runfin^{\mdp,\strat}(s)$, it holds
$$\probm^{\mdp,\strat}_{s}(\mathrm{Cyl}(\sigma))=\prod_{i=0}^{n-1}\prob(s_i,\strat(\sigma^i),s_{i+1})$$
and $\probm^{\mdp,\strat}_{s}(\mathrm{Cyl}(s))=1$. Intuitively, given a subset $X\subseteq \run^{\mdp,\strat}(s)$, $\probm^{\mdp,\strat}_{s}(X)$ is the probability that a run of $\mdp$ under $\strat$ that starts in $s$ belongs to the set $X$.


The following properties hold for any MDP $\mdp$ (see, \eg~\cite{baierBook}). For every EC $\mec$ of $\mdp$, there exists a finite-memory strategy $\strat$ for $\mdp$ such that $\mdp$ under $\strat$ starting from any state of $\mec$ never visits a state outside $\mec$ and all states of $\mec$ are visited infinitely many times with probability 1. 
On the other hand, having any, finite-memory or not, strategy $\strat$, a state $s$ of $\mdp$ and a run $\rho$ of $\mdp$ under $\strat$ that starts in $s$, the set of states visited infinitely many times by $\rho$ forms an end component. 
Let $\mathsf{ec}\subseteq \EC(\mdp)$ be the set of all ECs of $\mdp$ that correspond, in the above sense, to at least one run under the strategy $\strat$ that starts in the state $s$. 
We say that the strategy $\strat$ leads $\mdp$ from the state $s$ to the set $\mathsf{ec}$.

%
%

\subsection{Linear Temporal Logic}\label{subsec:ltl}

\begin{definition}
\emph{Linear Temporal Logic} (LTL) formulae over a set $\ap$ of atomic propositions are formed according to the following grammar:
\begin{equation*}
\phi::= true\mid a \mid \neg \phi \mid 
\phi \wedge \phi \mid \X \, \phi \mid \phi \, \U \, \phi \mid \G\, \phi \mid \F\, \phi,
\end{equation*}
where $a\in \ap$ is an atomic proposition, $\neg$ 
and $\wedge$ are standard Boolean connectives, and $\X$ (\emph{next}), $\U$ (\emph{until}), $\G$ (\emph{always}), and $\F$ (\emph{eventually}) are temporal operators.
\end{definition}

Formulae of LTL are interpreted over the words from $(2^{\ap})^{\omega}$, such as those induced by runs of an MDP $\mdp$ (for details see \eg \cite{baierBook}). For example, a word $w \in {(2^{\ap})}^\omega$ satisfies $\G \,\phi$ and $\F \, \phi$ if $\phi$ holds in $w$ always and eventually, respectively. If the word induced by a run $\rho \in \run^{\mdp}$ satisfies a formula $\phi$, we say that the run $\rho$ satisfies $\phi$. With slight abuse of notation, we also use states or sets of states of the MDP as propositions in LTL formulae. 


For every LTL formula $\phi$, the set of all runs of $\mdp$ that satisfy $\phi$ is measurable in the probability measure $\probm^{\mdp,\strat}_s$ for any $\strat$ and $s$ 
~\cite{baierBook}. 
With slight abuse of notation, we use LTL formulae as arguments of $\probm^{\mdp,\strat}_s$. 
If for a state $s\in S$ it holds that $\probm^{\mdp,\strat}_s(\phi)=1$, we say that the strategy $\strat$ almost-surely satisfies $\phi$ starting from $s$. If $\mdp$ is an initialized MDP and $\probm^{\mdp,\strat}_{s_{init}}(\phi)=1$, we say that $\strat$ almost-surely satisfies $\phi$.



The LTL control synthesis problem for an initialized MDP $\mdp$ and an LTL formula $\phi$ over $\ap$ aims to find a strategy for $\mdp$ that almost-surely satisfies $\phi$. This problem can be solved using 
principles from probabilistic model checking~\cite{baierBook},~\cite{dennismdpltlmaxprob}. The algorithm itself is based on the translation of $\phi$ to a Rabin automaton and the analysis of an MDP that combines the Rabin automaton and the original MDP $\mdp$. 

%


\begin{definition}
A \emph{deterministic Rabin automaton} (DRA) is a tuple $\ra=(Q,2^{\ap},\delta,q_{0},Acc)$, where $Q$ is a non-empty finite set of states, $2^{\ap}$ is an alphabet, $\delta \colon Q\times 2^{\ap} \to Q$ is a transition function, $q_{0}\in Q$ is an initial state, and $Acc\subseteq 2^{Q}\times 2^{Q}$ is an 
accepting condition.
\end{definition}


A run of $\ra$ is a sequence $q_{0}q_{1}\ldots \in Q^{\omega}$ such that for every $i\geq 0$, there exists $A_i\in 2^{\ap}$, $\delta(q_i,A_i)=q_{i+1}$. We say that the word $A_{0}A_{1}\ldots \in (2^{\ap})^{\omega}$ induces the run $q_{0}q_{1}\ldots $. 
A run of $\ra$ is called accepting if there exists a 
pair $(B,G)\in Acc$ such that the run visits every state from $B$ only finitely many times and at least one state from $G$ infinitely many times. 

For every LTL formula $\phi$ over $\ap$, there exists a DRA $\ra_{\phi}$ such that all and only words from $(2^{AP})^{\omega}$ satisfying $\phi$ induce an accepting run of $\ra_{\phi}$~\cite{ltldraProof}. For translation algorithms  
see \eg~\cite{ltldraAlgorithms}, and their online implementations, \eg~\cite{ltl2dstar}. 


\begin{definition}
Let $\mdp=(S,\act,\prob,\ap,L,g)$ be an initialized MDP and $\ra=(Q,2^{\ap},\delta,q_{0},Acc)$ be a DRA. The \emph{product} of $\mdp$ and $\ra$ is the initialized MDP $\product = (S_{\product}, \act, \prob_{\product},$ $\ap_{\product}, L_{\product},g_{\product})$, where $S_{\product}=S\times Q$, $\prob_{\product}((s,q),\alpha,(s',q'))=\prob(s,\alpha,s')$ if $q'=\delta(q,L(s))$ and $0$ otherwise, 
$\ap_{\product}=Q$, $L_{\product}((s,q))=q$, $g_{\product}((s,q),\alpha)=g(s,\alpha)$. The initial state of $\product$ is $s_{\product init}=(s_{init},q_{0})$.
\end{definition}

Using the projection on the first component, every (finite) run 
of $\product$ projects to a (finite) run of $\mdp$ and vice versa, for every (finite) run of $\mdp$, there exists a (finite) run of $\product$ that projects to it.
Analogous correspondence exists between strategies for $\product$ and $\mdp$. It holds that the projection of a finite-memory strategy for $\product$ is also finite-memory. More importantly, for the product $\product$ of an MDP $\mdp$ and a DRA $\ra_{\phi}$ for an LTL formula $\phi$, the probability of satisfying the accepting condition $Acc$ of $\ra_{\phi}$ under a strategy $\strat_{\product}$ for $\product$ starting from the initial state $s_{\product init}$, \ie
$$\probm^{\product, \strat_{\product}}_{s_{\product init}}\big(\bigvee \limits_{(B,G)\in Acc} (\F\G(\neg B)\: \wedge \: \G\F\, G)\big),$$
is equal to the probability of satisfying the formula $\phi$ in the MDP $\mdp$ under the projected strategy $\strat$ starting from the initial state $s_{init}$. 

\begin{definition}
Let $\product= (S_{\product}, \act, \prob_{\product}, \ap_{\product}, L_{\product},g_{\product})$ be the product of an MDP $\mdp$ and a DRA $\ra$. An \emph{accepting end component} (AEC) of $\product$ is defined as an end component $\mec=(S_{\mec},\act_{\mec},\prob_{\product},\ap_{\product},L_{\product},g_{\product})$ of $\product$ for which there exists a 
pair $(B,G)$ in the acceptance condition of $\ra$ such that $L_{\product}(S_{\mec})\cap B=\emptyset$ and $L_{\product}(S_{\mec})\cap G\neq \emptyset$. We say that $\mec$ is accepting with respect to the pair $(B,G)$. 
An AEC $\mec=(S_{\mec},\act_{\mec},\prob_{\product},\ap_{\product},L_{\product},g_{\product})$ is called \emph{maximal} (MAEC)
 if there is no AEC $\mec'=(S_{\mec'},\act_{\mec'},\prob_{\product},\ap_{\product},L_{\product},g_{\product})$ such that $\mec'\neq \mec$, $S_{\mec}\subseteq S_{\mec'}$, $\act_{\mec}((s,q))\subseteq \act_{\mec'}((s,q))$ for every $(s,q)\in S_{\product}$ and $\mec$ and $\mec'$ are accepting with respect to the same pair.  
We use $\AEC(\product)$ and $\MAEC(\product)$ to denote the set of all accepting end components and maximal accepting end components of $\product$, respectively.
\end{definition}

Note that MAECs that are accepting with respect to the same pair are always disjoint. However, MAECs that are accepting with respect to different pairs can intersect.

From the discussion above it follows that 
a necessary condition for almost-sure satisfaction of the accepting condition $Acc$ by a strategy $\strat_{\product}$ for $\product$ is that there exists a set $\mathsf{maec}\subseteq \MAEC(\product)$ of MAECs such that $\strat_{\product}$ leads the product from the initial state to $\mathsf{maec}$.


\section{Problem Formulation}\label{sec:pf}

Consider an initialized MDP $\mdp=(S,\act,\prob,\ap,L,g)$ and a specification given as an LTL formula $\phi$ over $\ap$ of the form
\begin{equation}\label{eq:ltlsur}
\phi= \varphi \wedge \G\F\, \pisur,
\end{equation}
where $\pisur\in \ap$ is an atomic proposition and $\varphi$ is an LTL formula over $\ap$. 
Intuitively, a formula of such form states two partial goals -- mission goal $\varphi$ and surveillance goal $\G\F \, \pisur$. To satisfy the whole formula the system must accomplish the mission and visit the surveillance states $S_{\sur}=\{s\in S\mid \pisur \in L(s)\}$ infinitely many times. The motivation for this form of specification comes from applications in robotics, where persistent surveillance tasks are often a part of the specification. 
Note that the form in Eq.~(\ref{eq:ltlsur}) does not restrict the full LTL expressivity since every LTL formula $\phi_1$ can be translated into a formula $\phi_2$ of the form in Eq.~(\ref{eq:ltlsur}) that is associated with the same set of runs of $\mdp$. Explicitly, $\phi_2=\phi_1 \wedge \G\F\, \pisur$, where $\pisur$ is such that $\pisur \in L(s)$ for every state $s\in S$. 

In this work, we focus on a control synthesis problem, 
where the goal is to almost-surely satisfy a given LTL specification, while optimizing a long-term quantitative objective. The objective is to minimize the average expected cumulative cost between consecutive visits to surveillance states.

Formally, we say that every visit to a surveillance state completes a surveillance cycle. In particular, starting from the initial state, the first visit to $S_{\sur}$ completes the first surveillance cycle of a run. We use $\cycles(\sigma)$ to denote the number of completed surveillance cycles in a finite run $\sigma$ plus one. 
For a strategy $\strat$ for $\mdp$, the 
cumulative cost in the first $n$ stages of applying $\strat$ to $\mdp$ starting from a state $s\in S$ is
$$g_{\mdp,\strat}(s,n)=\sum_{i=0}^{n}g(\sigma^{\mdp,\strat}_{s,n}(i),\strat({\sigma^{\mdp,\strat}_{s,n}}^{(i)})),$$
where $\sigma^{\mdp,\strat}_{s,n}$ is the random variable whose values are finite runs of length $n+1$ from the set $\runfin^{\mdp,\strat}(s)$ and the probability of a finite run $\sigma$ is $\probm^{\mdp,\strat}_{s}(\mathrm{Cyl}(\sigma))$. Note that $g_{\mdp,\strat}(s,n)$ is also a random variable.
Finally, we define the average expected cumulative cost per surveillance cycle (ACPC) in the MDP $\mdp$ under a strategy $\strat$ as a function $V_{\mdp,\strat}\colon S\to \mathbb{R}^+_0$
 such that for a state $s\in S$
\begin{equation*}\label{eq:acpc}
V_{\mdp,\strat}(s)=\limsup_{n\to \infty} E\Big(\frac{g_{\mdp,\strat}(s,n)}{\cycles(\sigma^{\mdp,\strat}_{s,n})}\Big).
\end{equation*}
The problem we consider in this paper can be formally stated as follows.

\begin{problem}\label{pf:acpc}
Let $\mdp=(S,\act,\prob,\ap,L,g)$ be an initialized MDP and $\phi$ be an LTL formula over $\ap$ of the form in Eq.~(\ref{eq:ltlsur}). Find a strategy $\strat$ for $\mdp$ such that $\strat$ almost-surely satisfies $\phi$ and, at the same time, $\strat$ minimizes the ACPC value $V_{\mdp,\strat}(s_{init})$ among all strategies almost-surely satisfying $\phi$.
\end{problem}

The above problem was recently investigated in~\cite{dennisCDC11}. However, the solution presented by the authors is guaranteed to find an optimal strategy only if every MAEC $\mec$ of the product $\product$ of the MDP $\mdp$ and the DRA for the specification satisfies certain conditions 
(for details see~\cite{dennisCDC11}). 
In this paper, we present a solution to Problem~\ref{pf:acpc} that always finds an optimal strategy 
if one exists. The algorithm is based on principles from probabilistic model checking~\cite{baierBook} and game theory~\cite{chatterjeeMFCS11}, whereas the authors in~\cite{dennisCDC11} mainly use results from dynamic programming~\cite{bertsekasVolumeII}.

In the special case when every state of $\mdp$ is a surveillance state, 
Problem~\ref{pf:acpc} aims to find a strategy that minimizes the average expected cost per stage among all strategies almost-surely satisfying $\phi$. The problem of minimizing the average expected cost per stage (ACPS) in an MDP, without considering any correctness specification, is a well studied problem in optimal control, see \eg~\cite{bertsekasVolumeII}. It holds that there always exists a stationary strategy that minimizes the ACPS value starting from the initial state. 
In our approach to Problem~\ref{pf:acpc}, we use techniques for solving the ACPS problem to find a strategy that minimizes the ACPC value.



\section{Solution}\label{sec:solution}

Let $\mdp=(S,\act,\prob,\ap,L,g)$ be an initialized MDP and $\phi$ an LTL formula over $\ap$ of the form in Eq.~(\ref{eq:ltlsur}). To solve Problem~\ref{pf:acpc} for $\mdp$ and $\phi$ we leverage ideas from game theory~\cite{chatterjeeMFCS11} and construct an optimal strategy for $\mdp$ as a combination of a strategy that ensures the almost-sure satisfaction of the specification $\phi$ and a strategy that guarantees the minimum ACPC value among all strategies that do not cause immediate unrepairable violation of $\phi$.

The algorithm we present in this section works with the product $\product=(S_{\product}, \act, \prob_{\product}, \ap_{\product}, L_{\product},g_{\product})$ of the MDP $\mdp$ and a deterministic Rabin automaton $\ra_{\phi} = (Q,2^{\ap},\delta,q_{0},Acc)$ for the formula $\phi$. We inherit the notion of a surveillance cycle in $\product$ by adding the proposition $\pisur$ to the set 
$\ap_{\product}$ and to the set $L_{\product}((s,q))$ for every $(s,q)\in S_{\product}$ such that $\pisur \in L(s)$. Using the correspondence between strategies for $\product$ and $\mdp$, an optimal strategy $\strat$ for $\mdp$ is found as a projection of a strategy $\strat_{\product}$ for $\product$ which almost-surely satisfies the accepting condition $Acc$ of $\ra_{\phi}$ and at the same time, minimizes the ACPC value $V_{\product,\strat_{\product}}(s_{\product init})$ among all strategies for $\product$ that almost-surely satisfy $Acc$.
 
Since $\strat_{\product}$ must almost-surely satisfy the accepting condition $Acc$, it leads from the initial state of $\product$ to a set of MAECs. 
For every MAEC $\mec$, the minimum ACPC value $V_{\mec}^{*}((s,q))$ that can be obtained in $\mec$ starting from a state $(s,q)\in S_{\mec}$ is equal for all the states of $\mec$ and we denote this value $V_{\mec}^*$. 
The strategy $\strat_{\product}$ is constructed in two steps. 

First, we find a set $\maec$ of MAECs of $\product$ and a strategy $\strat_0$ that leads $\product$ from the initial state to the set $\maec$. 
We require that $\strat_0$ and $\maec$ minimize the weighted average of the values $V_{\mec}^*$ for $\mec \in \maec$. 
The strategy $\strat_{\product}$ applies $\strat_0$ from the initial state until $\product$ enters the set $\maec$. 

Second, we solve the problem of how to control the product once a state of an MAEC $\mec\in \maec$ is visited. Intuitively, we combine two finite-memory strategies, $\strat_{\mec}^{\phi}$ for the almost-sure satisfaction of the accepting condition $Acc$ and $\strat_{\mec}^V$ for maintaining the average expected cumulative cost per surveillance cycle. To satisfy both objectives, the strategy $\strat_{\product}$ is played in rounds. In each round, we first apply the strategy $\strat_{\mec}^{\phi}$ and then the strategy $\strat_{\mec}^V$, each for a specific (finite) number of steps. 


\subsection{Finding an optimal set of MAECs}\label{subsec:maecstrat0}

Let $\MAEC(\product)$ be the set of all MAECs of the product $\product$ that can be computed as follows. For every pair $(B,G)\in Acc$, we create a new MDP from $\product$ by removing all its states with label in $B$ and the corresponding actions. For the new MDP, we use one of the algorithms in~\cite{deAlfaroThesis97}, \cite{yannakakis95}, \cite{chatterjeeFasterEC11} to compute the set of all its MECs. Finally, for every MEC, we check whether it contains a state with label in $G$.

In this section, the aim is to find a set $\maec\subseteq \MAEC(\product)$ and a strategy $\strat_0$ for $\product$ that satisfy conditions formally stated below. Since the strategy $\strat_0$ will only be used 
to enter the set $\maec$, it is constructed as a partial function.

\begin{definition}
A \emph{partial 
strategy} $\pstrat$ for the MDP $\mdp$ is a partial function $\pstrat\colon \runfin^{\product}\to \act$, where if $\pstrat(\sigma)$ is defined for $\sigma\in \runfin^{\product}$, then $\pstrat(\sigma)\in \act(last(\sigma))$. 
\end{definition}

A partial stationary strategy for $\mdp$ can also be considered as a partial function $\pstrat\colon S\to \act$ or a subset $\pstrat \subseteq S\times \act$. 
The set $\run^{\mdp,\pstrat}$ of runs of $\mdp$ under $\pstrat$ contains all infinite runs of $\mdp$ that follow $\pstrat$ and all those finite runs $\sigma$ of $\mdp$ under $\pstrat$ for which $\pstrat(last(\sigma))$ is not defined. A finite run of $\mdp$ under $\pstrat$ is then a finite prefix of a run under $\pstrat$. 
The probability measure $\probm^{\mdp,\pstrat}_{s}$ is defined in the same manner as in Sec.~\ref{prelim:mdp}. We also extend the semantics of LTL formulas to finite words. 
For example, a formula $\F\G\, \phi$ is satisfied by a finite word if in some non-empty suffix of the word $\phi$ always holds.  

The conditions on $\maec$ and $\strat_0$ are as follows. First, the partial strategy $\strat_0$ leads $\product$ 
to the set $\maec$, \ie 
\begin{equation}\label{eq:cond1}
\probm^{\product,\strat_{0}}_{s_{\product init}}(\F\G \, (\bigcup_{\mec \in \maec}S_{\mec}))=1.
\end{equation}
Second, we require that $\maec$ and $\strat_0$ minimize the value
\begin{equation}\label{eq:cond2}
\sum_{\mec\in \maec}\probm^{\product,\strat_{0}}_{s_{\product init}}(\F\G \, S_{\mec})\cdot V_{\mec}^{*}.
\end{equation}

The procedure to compute the optimal ACPC value $V_{\mec}^{*}$ for an MAEC $\mec$ of $\product$ is described in the next section. Assume we already computed this value for each MAEC of $\product$. The algorithm to find the set $\maec$ and partial strategy $\strat_0$ is based on an algorithm for stochastic shortest path (SSP) problem. The SSP problem is one of the basic optimization problems for MDPs. Given an initialized MDP and its state $t$, the goal is to find a strategy under which the MDP almost-surely reaches the state $t$, 
so called terminal state, while minimizing the expected cumulative cost. 
If there exists at least one strategy almost-surely reaching the terminal state, then there exists a stationary optimal strategy. For details and algorithms see \eg~\cite{bertsekasVolumeII}. 

The partial strategy $\strat_0$ and the set $\maec$ are computed as follows. First, we create a new MDP $\product'$ from $\product$ by considering only those states of $\product$ that can reach the set $\MAEC(\product)$ with probability 1 and their corresponding actions. The MDP $\product'$ can be computed using backward reachability from the set $\MAEC(\product)$. If $\product'$ does not contain the initial state $s_{\product init}$, there exists no solution to Problem~\ref{pf:acpc}. Otherwise, we add a new state $t$ and for every MAEC $\mec \in \MAEC(\product')=\MAEC(\product)$, we add a new action $\alpha_{\mec}$ to $\product'$. From each state $(s,q)\in S_{\mec}, \mec\in \MAEC(\product')$, we define a transition under $\alpha_{\mec}$ to $t$ with probability $1$ and set its cost to $V_{\mec}^*$. All other costs in the MDP are set to $0$. 
Finally, we solve the SSP problem for $\product'$ and the state $t$ as the terminal state. Let $\strat_{SSP}$ be the resulting stationary optimal strategy for $\product'$. For every $(s,q)\in S_{\product}$, we define $\strat_0((s,q))=\strat_{SSP}((s,q))$ if the action $\strat_{SSP}((s,q))$ does not lead from $(s,q)$ to $t$, $\strat_0((s,q))$ is undefined otherwise. The set $\maec$ is the set of all MAECs $\mec$ for which there exists a state $(s,q)$ such that $\strat_{SSP}((s,q))=\alpha_{\mec}$.

\begin{proposition}\label{prop:maecstrat0}
The set $\maec$ and the partial stationary strategy $\strat_0$ resulting from the above algorithm satisfy the conditions in Eq.~(\ref{eq:cond1}) and Eq.~(\ref{eq:cond2}).
\end{proposition}

\begin{proof}
%
Both conditions follow directly from the fact that the strategy $\strat_{SSP}$ is an optimal solution to the SSP problem for $\product'$ and $t$.
\end{proof}


\subsection{Optimizing ACPC value in an MAEC}\label{subsec:reduction}

In this section, we compute the minimum ACPC value $V_{\mec}^*$ that can be attained in an MAEC $\mec\in \MAEC(\product)$ and construct the corresponding strategy for $\mec$. 
Essentially, we reduce the problem of computing the minimum ACPC value to the problem of computing the minimum ACPS value 
by reducing $\mec$ to an MDP such that every state of the reduced MDP is labeled with the surveillance proposition $\pisur$.

%

Let $\mec=(S_{\mec},\act_{\mec},\prob_{\product},\ap_{\product},L_{\product},g_{\product})$ be an MAEC of $\product$. Since it is an MAEC, there exists a state $(s,q)\in S_{\mec}$ with $\pisur \in L_{\product}((s,q))$. Let $S_{\mec_{\sur}}$ denote the set of all such states in $S_\mec$. We reduce $\mec$ to an MDP 
$$\mec_{\sur}=(S_{\mec_{\sur}},\mathbf{\act}_{\sur},\prob_{\sur},\ap_{\product},L_{\product},g_{\sur})$$
using Alg.~\ref{alg:reduction}. For the sake of readability, we use singletons such as $v$ instead of pairs such as $(s,q)$ to denote the states of $\mec$. The MDP $\mec_{\sur}$ is constructed from $\mec$ by eliminating states from $S_{\mec}\backslash S_{\mec_{\sur}}$ one by one in arbitrary order.
The actions $\mathbf{\act}_{\sur}$ are partial stationary strategies for $\mec$ in which we remember all the states and actions we eliminated. Later we prove that the transition probability $\prob_{\sur}(v,\pstrat,v')$ for states $v,v'\in S_{\mec_{\sur}}$ and an action $\pstrat\in \mathbf{\act}_{\sur}(v)$ is the probability that in $\mec$ under the partial stationary strategy $\pstrat$, if we start from the state $v$, the next state that will be visited from the set $S_{\mec_{\sur}}$ is the state $v'$, \ie the first surveillance cycle is completed by visiting $v'$. The cost $g_{\sur}(v,\pstrat)$ is the expected cumulative cost gained in $\mec$ using partial stationary strategy $\pstrat$ from $v$ until we reach a state in $S_{\mec_{\sur}}$.

\begin{figure*}[t]
\begin{center}
\scalebox{0.4}{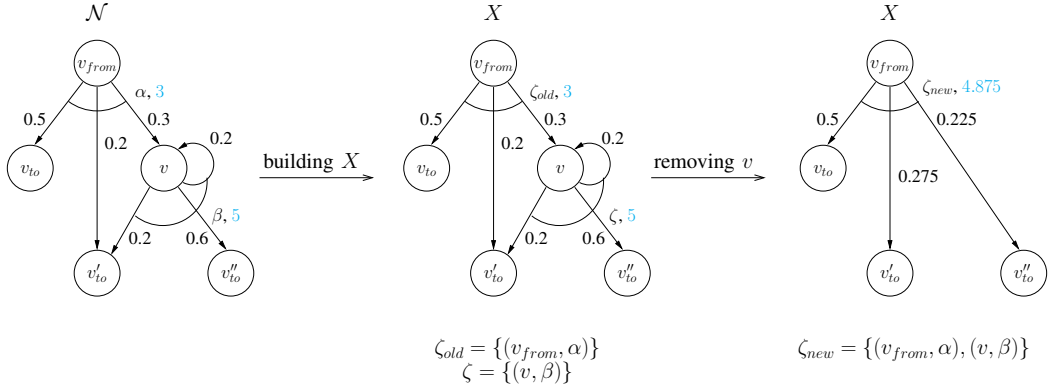}
\end{center}
\caption{Illustration of Alg.~\ref{alg:reduction}. A part of an MAEC $\mec$ is shown in the left. An auxiliary MDP $X$ is constructed by transforming actions of $\mec$ to partial stationary strategies. The MDP $X$ after eliminating the state $v$ is shown on the right. The costs associated with actions are depicted in blue.}
\label{fig:redex}
\end{figure*}

In Fig.~\ref{fig:redex}, we demonstrate the reduction on an example using the notation introduced in Alg.~\ref{alg:reduction}. 
On the left side, we see a part of an MAEC $\mec$ with five states and two actions. First, we build an MDP $X=(S_X,\mathbf{\act}_{X},\prob_{X},\ap_{\product},L_{\product},g_{X})$ from $\mec$ by transforming every action of every state to a partial stationary strategy with a single pair given by the state and the action. 
The MDP $X$ is used in the algorithm as an auxiliary MDP to store the current version of the reduced system. Assume we want to reduce the state $v$. We consider all ``incoming'' and ``outgoing'' actions of $v$ and combine them pairwise as follows. There is only one outgoing action from $v$ in $X$, namely $\pstrat$, and only one incoming action, namely action $\pstrat_{old}$ of state $v_{from}$. Since $\pstrat$ and $\pstrat_{old}$ do not conflict as partial stationary strategies on any state of $\mec$, we merge them to create a new partial stationary strategy $\pstrat_{new}$ that is an action of $v_{from}$. The transition probability $\prob_X(v_{from},\pstrat_{new},v_{to})$ for a state $v_{to}$ of $X$ is computed as the sum of the transition probability $\prob_X(v_{from},\pstrat_{old},v_{to})$ of transiting from $v_{from}$ to $v_{to}$ using the old action $\pstrat_{old}$ and the probability of entering $v_{to}$ by first transiting from $v_{from}$ to $v$ using $\pstrat_{old}$ and from $v$ eventually reaching $v_{to}$ using $\pstrat$. 
The cost $g_X(v_{from},\pstrat_{new})$ is the expected cumulative cost gained starting from $v_{from}$ by first applying action $\pstrat_{old}$ and if we transit to $v$, applying $\pstrat$ until a state different from $v$ is reached. Now that we considered every pair of an incoming and outgoing action of $v$, the state $v$ and its incoming and outgoing actions are reduced. The modified MDP $X$ is depicted on the right side of Fig.~\ref{fig:redex}. 

\begin{algorithm}[t]
\small
\caption{Reduction of an MAEC $\mec$ to $\mec_{\sur}$}
\label{alg:reduction}
\begin{algorithmic}[1]
\REQUIRE{$\mec=(S_{\mec},\act_{\mec},\prob_{\product},\ap_{\product},L_{\product},g_{\product})$}
\ENSURE{$\mec_{\sur}=(S_{\mec\sur},\mathbf{\act}_{\sur},\prob_{\sur},\ap_{\product},L_{\product},g_{\sur})$}

\STATE let $X=(S_X,\mathbf{\act}_{X},\prob_{X},\ap_{\product},L_{\product},g_{X})$ be an MDP where \\
\begin{itemize}
\item $S_X:= S_{\mec}$,
\item for $v\in S_X:$\\
$\mathbf{A}_{X}(v):=\{\pstrat_{\alpha}\mid \pstrat_{\alpha}=\{(v,\alpha)\}, \alpha \in \act_{\mec}(v)\}$,
\item for $v,v'\in S_X, \pstrat\in \mathbf{\act}_{X}:$\\
$\prob_{X}(v,\pstrat,v') := \prob_{\product}(v,\pstrat(v),v')$,
\item for $v\in S_X, \pstrat\in \mathbf{\act}_{X}:$\\
$g_{X}(v,\pstrat):=g_{\product}(v,\pstrat(v))$
\end{itemize}

\WHILE{$S_X\backslash S_{\mec\sur}\neq \emptyset$}
\STATE let $v\in S_X\backslash S_{\mec\sur}$
\FORALL{$\pstrat\in \mathbf{\act}_{X}(v)$}
\IF{$\prob_{X}(v,\pstrat,v)<1$}
\FORALL{$v_{from}\in S_X, \pstrat_{old}\in \mathbf{\act}_{X}(v_{from})$}
\IF{$\prob_{X}(v_{from},\pstrat_{old},v)>0$ and $\pstrat_{old},\pstrat$ do not conflict for any state from $S_X$}
\STATE $\pstrat_{new}:=\pstrat_{old} \cup \pstrat$
\STATE add $\pstrat_{new}$ to $\mathbf{\act}_{X}(v_{from})$
\STATE for every $v_{to}\in S_X$:\\
{\scriptsize
\begin{equation*}
\hspace*{-1.6cm}
\begin{split}
\prob_{X}(v_{from},\pstrat_{new},v_{to}) :=\, & \prob_{X}(v_{from},\pstrat_{old},v_{to})\, +\\
& +\prob_{X}(v_{from},\pstrat_{old},v)\cdot \frac{\prob_{X}(v,\pstrat,v_{to})}{1-\prob_{X}(v,\pstrat,v)}\\
g_{X}(v_{from},\pstrat_{new}) :=\, & g_{X}(v_{from},\pstrat_{old})\, + \\
& + \prob_{X}(v_{from},\pstrat_{old},v)\cdot \frac{g_{X}(v,\pstrat)}{1-\prob_{X}(v,\pstrat,v)}
\end{split}
\end{equation*}
}
\STATE remove $\pstrat_{old}$ from $\mathbf{\act}_{X}(v_{from})$
\ENDIF
\ENDFOR
\ENDIF
\STATE remove $\pstrat$ from $\mathbf{\act}_{X}(v)$
\ENDFOR
\STATE remove $v$ from $S_X$
\ENDWHILE
\STATE return $X$
\end{algorithmic}
\end{algorithm}

\begin{proposition}\label{prop:reduction}
Let $\mec=(S_{\mec},\act_{\mec},\prob_{\product},\ap_{\product},L_{\product},g_{\product})$ be an MAEC 
and $\mec_{\sur}=(S_{\mec_{\sur}},\mathbf{\act}_{\sur},\prob_{\sur},\ap_{\product},L_{\product},g_{\sur})$ its reduction resulting from Alg.~\ref{alg:reduction}. The minimum ACPC value that can be attained in $\mec_{\sur}$ starting from any of its states is the same and we denote it $V_{\mec_{\sur}}^*$. There exists a stationary strategy $\strat_{\mec_{\sur}}^V$ for $\mec_{\sur}$ that attains this value regardless of the starting state in $\mec_{\sur}$. Both $V_{\mec_{\sur}}^*$ and $\strat_{\mec_{\sur}}^V$ can be computed as a solution to the ACPS problem for $\mec_{\sur}$. 
It holds that $V_{\mec}^*=V_{\mec_{\sur}}^*$ and 
from $\strat_{\mec_{\sur}}^V$, one can construct a finite-memory strategy $\strat_{\mec}^V$ for $\mec$ which regardless of the starting state in $\mec$ attains the optimal ACPC value $V_{\mec}^*$.
\end{proposition}

\begin{proof}
We prove the following correspondence between $\mec$ and $\mec_{\sur}$. For every $v,v'\in S_{\mec_{\sur}}$ and $\pstrat\in \mathbf{\act}_{\sur}(v)$, it holds that $\pstrat$ is a well-defined partial stationary strategy for $\mec$. The transition probability $\prob_{\sur}(v,\pstrat,v')$ is the probability that in $\mec$, when applying $\pstrat$ starting from $v$, the first surveillance cycle is completed by visiting $v'$, \ie  
$$\prob_{\sur}(v,\pstrat,v')=\probm^{\mec,\pstrat}_{v}(\X (\neg S_{\mec \sur}\U \, v')).$$
The cost $g_{\sur}(v,\pstrat)$ is the expected cumulative cost gained in $\mec$ when applying $\pstrat$ starting from $v$ until 
the first surveillance cycle is completed. On the other hand, for every partial stationary strategy $\pstrat$ for $\mec$ such that 
$$\probm^{\mec,\pstrat}_{v}(\F\, S_{\mec_{\sur}})=1$$
for some $v\in S_{\mec \sur}$, there exists an action $\pstrat'\in \mathbf{\act}_{\sur}(v)$ such that the action $\pstrat'$ corresponds to the partial stationary strategy $\pstrat$ in the above sense, \ie
$$\prob_{\sur}(v,\pstrat',v')=\probm^{\mec,\pstrat}_{v}(\X (\neg S_{\mec_{\sur}} \U\, v'))$$
for every $v'\in S_{\mec_{\sur}}$, and the cost $g_{\sur}(v,\pstrat')$ is the expected cumulative cost gained in $\mec$ when we apply $\pstrat$ starting from $v$ until we reach a state in $S_{\mec_{\sur}}$.

To prove the first part of the correspondence above, we prove the following invariant of Alg.~\ref{alg:reduction}. Let $X=(S_X,\mathbf{\act}_{X},\prob_{X},\ap_{\product},L_{\product},g_{X})$ be the MDP from the algorithm after the initialization, before the first iteration of the while cycle. It is easy to see that all actions of $X$ are well-defined partial stationary strategies. For the transition probabilities, it holds that
$$\prob_{X}(v_{from},\pstrat,v_{to})=\probm^{\mec,\pstrat}_{v_{from}}(\X(\neg S_X \U\, v_{to}))$$
for every $v_{from},v_{to}\in S_X$ and $\pstrat \in \mathbf{\act}_{X}(v_{from})$. The cost $g_{X}(v_{from},\pstrat)$ is the expected cumulative cost gained in $\mec$ starting from $v_{from}$ when applying $\pstrat$ until we reach a state in $S_X$. We show that these conditions also hold after every iteration of the while cycle.

Let $X$ satisfy the conditions above and let $v\in S_X\backslash S_{\mec\sur}$. By removing the state $v$ from $S_X$, we obtain a new version of the MDP $X'=(S_{X'},\mathbf{\act}_{X'},\prob_{X'},\ap_{\product},L_{\product},g_{X'})$. Note that $S_{X'}\cup \{v\}=S_X$. Let $v_{from}\in S_{X'}$ be a state of $X'$ and $\pstrat_{new}\in \mathbf{\act}_{X'}(v_{from})$ be its action such that $\pstrat_{new}$ has changed in the process of removing the state $v$. The action $\pstrat_{new}$ is a well-defined partial stationary strategy because it must have been created as a union of an action $\pstrat_{old}$ of $v_{from}$ and an action $\pstrat$ of $v$, both from the previous version $X$, which do not conflict on any state from $S_X$.

Let $\stackrel{X'}{\rightarrow} v_{to}$ denote the LTL formula $\X(\neg S_{X'} \U\, v_{to})$. For a state $v_{to}\in S_{X'}$, we prove that
$$\prob_{X'}(v_{from},\pstrat_{new},v_{to})=\probm^{\mec,\pstrat_{new}}_{v_{from}}(\stackrel{X'}{\rightarrow} v_{to}).$$
Since $\pstrat_{new}=\pstrat_{old}\cup \pstrat$, the probability in $\mec$ when applying $\pstrat_{new}$ starting from $v_{from}$ of reaching the state $v_{to}$ as the next state in $S_{X'}$ is the probability of reaching it as the next state in $S_X$ when using $\pstrat_{old}$ from $v_{from}$, plus the probability of reaching $v$ as the next state in $S_X$ from $v_{from}$ using $\pstrat_{old}$ and then eventually reaching the state $v_{to}$ from $v$ using $\pstrat$. This means 
\begin{equation*}
\footnotesize
\begin{split}
\probm^{\mec,\pstrat_{new}}_{v_{from}}(\stackrel{X'}{\rightarrow} v_{to}) &= \probm^{\mec,\pstrat_{old}}_{v_{from}}(\stackrel{X}{\rightarrow} v_{to})\, +\\
& \quad + \probm^{\mec,\pstrat_{old}}_{v_{from}}(\stackrel{X}{\rightarrow} v)\cdot \probm^{\mec,\pstrat}_{v}(\F\, v_{to})\\
& = \prob_{X}(v_{from},\pstrat_{old},v_{to}) + \prob_{X}(v_{from},\pstrat_{old},v)\cdot \\
& \quad \cdot \big(\sum \limits_{i=0}^{\infty}\prob_X(v,\pstrat,v)^i \cdot \prob_X(v,\pstrat,v_{to})\big)\\
& = \prob_{X}(v_{from},\pstrat_{old},v_{to})\, +\\
& \quad +\prob_{X}(v_{from},\pstrat_{old},v)\cdot \frac{\prob_{X}(v,\pstrat,v_{to})}{1-\prob_{X}(v,\pstrat,v)}
\end{split}
\end{equation*}
which is exactly as defined in Alg.~\ref{alg:reduction}.

Similarly, we prove that $g_{X'}(v_{from},\pstrat_{new})$ is the expected cumulative cost gained in $\mec$ starting from $v_{from}$ when applying $\pstrat_{new}$ until we reach a state in $S_{X'}$. As $\pstrat_{new}=\pstrat_{old}\cup \pstrat$, it is the expected cumulative cost of reaching a state in $S_X$ by using $\pstrat_{old}$ plus, in the case we reach $v$, the expected cumulative cost of eventually reaching a state in $S_{X'}$, \ie other than $v$, using $\pstrat$. To be specific, we have
\begin{equation*}
\footnotesize
\begin{gathered}
 g_{X}(v_{from},\pstrat_{old})\, + \prob_{X}(v_{from},\pstrat_{old},v)\cdot \\
 \cdot \big(\sum \limits_{i=0}^{\infty}\prob_X(v,\pstrat,v)^i\cdot (1-\prob_X(v,\pstrat,v))\cdot (i+1)\cdot g_X(v,\pstrat)\\
 = \\
 g_{X}(v_{from},\pstrat_{old})\, + \prob_{X}(v_{from},\pstrat_{old},v)\cdot \frac{g_{X}(v,\pstrat)}{1-\prob_{X}(v,\pstrat,v)}, 
\end{gathered}
\end{equation*}
just as defined in Alg.~\ref{alg:reduction}. This completes the proof of the first part of the correspondence between $\mec$ and $\mec_{\sur}$.

The second part of the correspondence between $\mec$ and $\mec_{\sur}$ follows directly from the fact that, in the process of removing a state $v\in S_X\backslash S_{\mec_{\sur}}$, we consider all combinations of actions of $v$ which eventually reach a state different from $v$, with all actions of all states $v_{from}$ having an action under which $v$ is reached with non-zero probability.  

From the correspondence between $\mec$ and $\mec_{\sur}$ it follows that in $\mec_{\sur}$, there exists a finite run between every two states. Therefore, the minimum ACPC value that can be obtained in $\mec_{\sur}$ from any of its states is the same and it is denoted by $V_{\mec_{\sur}}^*$. Since every state of $\mec_{\sur}$ is a surveillance state, the ACPC problem for $\mec_{\sur}$ is equivalent to solving the ACPS problem for $\mec_{\sur}$. Using one of the algorithms in~\cite{bertsekasVolumeII}, we obtain a stationary strategy $\strat_{\mec_{\sur}}^V$ that attains the ACPC value $V_{\mec_{\sur}}^*$ regardless of the starting state. From the correspondence between $\mec$ and $\mec_{\sur}$ it also follows that $V_{\mec_{\sur}}^*=V_{\mec}^*$.

Now we construct the strategy $\strat_{\mec}^V$ for $\mec$ 
 and show that it attains the minimum ACPC value $V_{\mec}^*$ regardless of the initial state. Intuitively, the strategy $\strat_{\mec}^V$ is constructed to lead to a single EC of $\mec$ that provides the minimum ACPC value and that is the EC encoded by the strategy $\strat_{\mec_{\sur}}^V$ for $\mec_{\sur}$.

Let $S_{def}\subseteq S_{\mec}$ be the set of all states $v\in S_\mec$ for which there exists a surveillance state $v_{\sur}\in S_{\mec_{\sur}}$ such that the partial strategy $\strat_{\mec_{\sur}}^{V}(v_{\sur})$ for $\mec$ is defined on the state $v$. We compute a partial strategy $\pstrat_{init}$ that leads from every state from $S_{\mec}\backslash S_{def}$ to the set $S_{def}$ as follows. Let $\mec'$ be an MDP that is created from $\mec$ by adding a new state $t$ and a new action $\alpha_{def}$. From every state  $v\in S_{def}$, we define a new transition under $\alpha_{def}$ to $t$ with probability 1 and cost 0. Let $\strat_{SSP}$ be a stationary optimal strategy for the SSP problem for $\mec'$ and $t$ as the terminal state. We define $\pstrat_{init}(v) = \strat_{SSP}(v)$ for every $v\in S_{\mec}\backslash S_{def}$.

The strategy $\strat_{\mec}^{V}$ is a then finite-memory strategy 
$$\strat_{\mec}^{V} = (M,\mathsf{act},\Delta,\mathsf{start}),$$
where $M=S_{\mec_{\sur}}\cup \{init\}$ is the set of modes, $\Delta\colon M\times S_{\mec}\to M$ is the transition function such that for every $m\in M,v\in S_{\mec}$
$$\Delta(m,v)=\begin{cases}
m & \text{if }v\not \in S_{\mec\sur},\\
v & \text{otherwise}.
\end{cases}$$
The function $\mathsf{act}\colon M\times S_{\mec}\to \act_{\mec}$ that selects an action to be applied in $\mec$ is for $m\in M,v\in S_{\mec}$ defined as
$$\mathsf{act}(m,v)=
\begin{cases}
\big(\strat_{\mec_{\sur}}^{V}(m)\big)(v) & \text{if }m\in S_{\mec_{\sur}}\\
\pstrat_{init}(v) & \text{otherwise.}
\end{cases}$$
Finally, $\mathsf{start}\colon S_{\mec}\to S_{\mec_{\sur}}$ selecting the starting mode for $v\in S_{\mec}$ is defined as 
$$\mathsf{start}(v)=\begin{cases}
v & \text{if }v\in S_{\mec_{\sur}},\\
m & \text{where } \big(\strat_{\mec_{\sur}}^{V}(m)\big)(v)\\
& \text{is defined,}\\
init & \text{otherwise.}
\end{cases}$$
The strategy attains the ACPC value $V_{\mec}^*$ since it only simulates the strategy $\strat_{\mec_{\sur}}^V$ by unwrapping the corresponding partial strategies.  
\end{proof}

The following property of the strategy $\strat_{\mec}^{V}$ is crucial for the correctness of our approach to Problem~\ref{pf:acpc}.

\begin{proposition}\label{prop:stratv}
For every $(s,q)\in S_{\mec}$, it holds that
$$\lim_{n\to \infty} \probm^{\mec,\strat_{\mec}^{V}}_{(s,q)}(\{\rho \mid \frac{g_{\product}(\rho^{(\cycles n)})}{n}\leq V_{\mec}^{*}\})=1,$$
where $g_{\product}(\rho^{(\cycles n)})$ denotes the cumulative cost gained in the first $n$ surveillance cycles of a run $\rho \in \run^{\mec}((s,q))$. Hence, for every $\epsilon>0$, there exists $j(\epsilon)\in \mathbb{N}$ such that if the strategy $\strat_{\mec}^{V}$ is applied from a state $(s,q)\in S_{\mec}$ for any $l\geq j(\epsilon)$ surveillance cycles, then the average expected cumulative cost per surveillance cycle in these $l$ surveillance cycles is at most $V_{\mec}^{*}+\epsilon$ with probability at least $1-\epsilon$, \ie
\begin{equation*}
\probm^{\mec,\strat_{\mec}^{V}}_{(s,q)}(\{\rho \mid \frac{g_{\product}(\rho^{(\cycles l)})}{l}\leq V_{\mec}^{*}+\epsilon\})\geq 1-\epsilon.
\end{equation*}
\end{proposition}

\medskip

\begin{proof}
In~\cite{chatterjeeFasterEC11} the authors prove that a strategy solving the ACPS problem for an MDP satisfies a property analogous to the one in the proposition. Especially, for the strategy $\strat_{\mec_{\sur}}^V$ for the reduced MDP $\mec_{\sur}$, it holds that for any state $(s,q)\in S_{\mec_\sur}$
$$\lim_{n\to \infty} \probm^{\mec_{\sur},\strat_{\mec_{\sur}}^{V}}_{(s,q)}(\{\rho \mid \frac{g_{\mec_{\sur}}(\rho^{(n)})}{n}\leq V_{\mec_{\sur}}^{*}\})=1,$$
where $g_{\mec_{\sur}}(\rho^{(n)})$ denotes the cumulative cost gained in the first $n$ stages of a run $\rho \in \run^{\mec_{\sur}}((s,q))$. The proposition then follows directly from the construction of the strategy $\strat_{\mec}^V$ from the strategy $\strat_{\mec_{\sur}}^V$. 
\end{proof}


\subsection{Almost-sure acceptance in an MAEC}\label{subsec:accinmaec}

Here we design a strategy for an MAEC $\mec\in \MAEC(\product)$ that guarantees almost-sure satisfaction of the acceptance condition $Acc$ of $\ra_{\phi}$. Let $(B,G)$ be a pair in $Acc$ such that $\mec$ is accepting with respect to $(B,G)$, \ie $L_{\product}(S_{\mec})\cap B=\emptyset$ and $L_{\product}(S_{\mec})\cap G\neq \emptyset$. There exists a stationary strategy $\strat_{\mec}^{\phi}$ for $\mec$ under which a state with label in $G$ is reached with probability 1 regardless of the starting state, \ie 
\begin{equation}\label{eq:stratphi}
\probm^{\mec, \strat_{\mec}^{\phi}}_{(s,q)}(\F\, G)=1
\end{equation}
for every $(s,q)\in S_{\mec}$. The existence of such a strategy follows from the fact that $\mec$ is an EC~\cite{baierBook}. Moreover, we construct $\strat_{\mec}^{\phi}$ to minimize the expected cumulative cost before reaching a state in $S_{\mec}\cap S\times G$.

The strategy $\strat_{\mec}^{\phi}$ is found as follows. Let $\mec'$ be an MDP that is created from $\mec$ by adding a new state $t$ and a new action $\alpha_G$. From every state  $(s,q)\in S_{\mec}\cap S\times G$, we define a new transition under $\alpha_G$ to $t$ with probability 1 and cost 0. Let $\strat_{SSP}$ be a stationary optimal strategy for the SSP problem for $\mec'$ and $t$ as the terminal state. For a state $(s,q)\in S_{\mec}$, we define $\strat_{\mec}^{\phi}((s,q))=\strat_{SSP}((s,q))$ if the state $(s,q)$ does not have a label in $G$, otherwise $\strat_{\mec}^{\phi}((s,q))=\alpha$ for some $\alpha\in \act_{\mec}((s,q))$.

\begin{proposition}\label{prop:stratacc}
The strategy $\strat_{\mec}^{\phi}$ for $\mec$ resulting from the above algorithm almost-surely reaches the set $S_{\mec}\cap S\times G$ and minimizes the expected cumulative cost before reaching the set, regardless of the initial state.
\end{proposition}

\begin{proof}
It follows directly from the fact that $\strat_{SSP}$ optimally solves the SSP problem for the MDP $\mec'$ and $t$. 
\end{proof}

\subsection{Optimal strategy for $\product$}\label{subsec:stratproduct} 

Finally, we are ready to construct the strategy $\strat_{\product}$ for the product $\product$ that projects to an optimal solution for $\mdp$. 

First, starting from the initial state $s_{\product init}$, $\strat_{\product}$ applies the strategy $\strat_0$ resulting from the algorithm described in Sec.~\ref{subsec:maecstrat0} 
until a state 
of an MAEC in the set $\maec$ is reached. Let $\mec\in \maec$ denote the MAEC and let $(B,G)\in Acc$ be a pair from the accepting condition of $\ra_{\phi}$ such that $\mec$ is accepting with respect to $(B,G)$. 

Now, the strategy $\strat_{\product}$ starts to play the rounds. Each round consists of two phases. First, play the strategy $\strat_{\mec}^{\phi}$ from Sec.~\ref{subsec:accinmaec} until a state with label in $G$ is reached. Let us denote $k_{i}$ the number of steps we play $\strat_{\mec}^{\phi}$ in $i$-th round. The second phase applies the strategy $\strat_{\mec}^{V}$ from Sec.~\ref{subsec:reduction} until the number of completed surveillance cycles in the second phase of the current round is $l_{i}$. The number $l_{i}$ is any natural number for which 
$$l_{i}\geq \max\{j(\tfrac{1}{i}),i\cdot k_{i}\cdot g_{\product max}\},$$
where $j(\frac{1}{i})$ is from Prop.~\ref{prop:stratv} and $g_{\product max}$ is the maximum value of the costs $g_{\product}$. After applying the strategy $\strat_{\mec}^{V}$ for $l_i$ surveillance cycles, we proceed to the next round $i+1$.


\begin{theorem}\label{prop:corr}
The strategy $\strat_{\product}$ almost-surely satisfies the accepting condition $Acc$ of $\ra_{\phi}$ and at the same time, $\strat_{\product}$ minimizes the ACPC value $V_{\product, \strat_{\product}}(s_{\product init})$ among all strategies for $\product$ almost-surely satisfying $Acc$.
\end{theorem}

\begin{proof}
From Prop.~\ref{prop:maecstrat0} it follows that when applying the strategy $\strat_0$ from the initial state $s_{\product init}$, the set $\maec$ is reached with probability 1. 

Assume that $\product$ enters MAEC $\mec\in \maec$ that is accepting with respect to a pair $(B,G)\in Acc$. Let $i$ be the current round of $\strat_{\product}$ and $\epsilon_i=\tfrac{1}{i}$. According to Prop.~\ref{prop:stratacc}, a state with a label in $G$ is almost-surely reached. In addition, using Prop.~\ref{prop:stratv}, the average expected cumulative cost per surveillance cycle in the $i$-th round is at most 
\begin{align*}
&\frac{k_{i}\cdot g_{\mec max} + l_{i}(V_{\mec}^{*}+\epsilon_{i})}{l_{i}} = \\
&\qquad \qquad = V_{\mec}^{*} + \epsilon_{i} + \frac{k_{i}\cdot g_{\mec max}}{l_{i}}\\
&\qquad \qquad \leq V_{\mec}^{*} + \epsilon_{i} + \frac{1}{i} \qquad \qquad (l_{i}\geq i\cdot k_{i}\cdot g_{\mec max})\\
&\qquad \qquad = V_{\mec}^{*} + \frac{2}{i}
\end{align*}
with probability at least $1-\frac{1}{i}$. Therefore, in the limit, in the MAEC $\mec$, we both satisfy the LTL specification and reach the optimal ACPC value with probability $1$. Together with the fact that $\maec$ and $\strat_0$ satisfy the condition in Eq.~(\ref{eq:cond2}), we have that $\strat_{\product}$ is an optimal strategy for $\product$.
\end{proof}

\subsection{Complexity and discussion}\label{subsec:compldisc}

The size of a Rabin automaton for an LTL formula $\phi$ is in the worst case doubly exponential in the size of the set $\ap$. However, studies such as \cite{ltldraAlgorithms} show that in practice, for many LTL formulas, automata are much smaller and manageable. 

Once the product $\product$ is built, we compute the set $\MAEC(\product)$ by running $|Acc|$-times an algorithm for MEC decomposition, which is polynomial in the size of $\product$. The size of the set $\MAEC(\product)$ is in the worst case $|Acc|\cdot |S_\product|$. For each MAEC $\mec$, we compute its reduction $\mec_{\sur}$ using Alg.~\ref{alg:reduction} in time $\mathcal{O}(|S_{\mec}|\cdot |\act_\mec|^{\mathcal{O}(|S_\mec|)})$. The optimal ACPC value $V_{\mec}^*$ and an optimal finite-memory strategy $\strat_{\mec}^V$ are then found in time polynomial in the size of the reduced MDP.

\begin{figure*}[t]
\centering
\begin{tabular}{c c}
\scalebox{0.5}{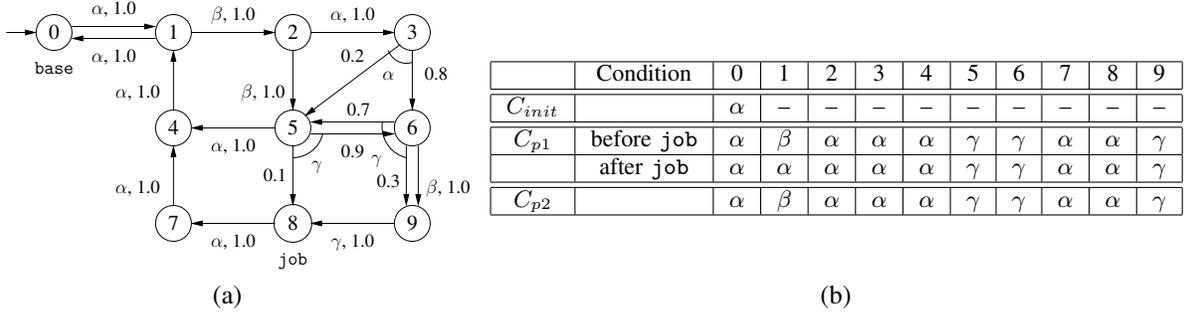}
&
\raisebox{1.7cm}{
\small
\begin{tabular}{| c | c | c | c | c | c | c | c | c | c | c | c |}
\hline 
& Condition & 0 & 1 & 2 & 3 & 4 & 5 & 6 & 7 & 8 & 9 \\
\hline 
\hline 
$\strat_{init}$ & & $\alpha$ & -- & -- & -- & -- & -- & -- & -- & -- & --\\
\hline 
\hline 
$\strat_{p1}$ & before $\mathtt{job}$ & $\alpha$ & $\beta$ & $\alpha$ & $\alpha$ & $\alpha$ & $\gamma$ & $\gamma$ & $\alpha$ & $\alpha$ & $\gamma$ \\
\hline
& after $\mathtt{job}$ & $\alpha$ & $\alpha$ & $\alpha$ & $\alpha$ & $\alpha$ & $\gamma$ & $\gamma$ & $\alpha$ & $\alpha$ & $\gamma$ \\
\hline
\hline
$\strat_{p2}$ & & $\alpha$ & $\beta$ & $\alpha$ & $\alpha$ & $\alpha$ & $\gamma$ & $\gamma$ & $\alpha$ & $\alpha$ & $\gamma$ \\
\hline
\end{tabular}
}
\\
(a) & (b)
\end{tabular}
\caption{(a) Initialized MDP $\mdp$ with initial state 0. 
The costs of applying $\alpha,\beta,\gamma$ in any states are 5, 10, 1, respectively, \eg $g(1,\alpha)=5$. (b)~Definitions of strategies $\strat_{init},\strat_{p1},\strat_{p2}$ for $\mdp$, the projections of strategies $\strat_0,\strat_{\mec}^\phi$, $\strat_{\mec}^V$ for $\product$, respectively. The condition ``before $\mathtt{job}$'' means that the corresponding prescription is used if the job location has not yet been visited since the last visit of the base. Similarly, the prescription with condition ``after $\mathtt{job}$'' is used if the job location was visited at least once since the last visit of the base.}
\label{fig:casestudy}
\end{figure*}

The algorithm for finding the strategy $\strat_0$ and the optimal set $\maec$ are again polynomial in the size of $\product$. Similarly, computing a stationary strategy $\strat_{\mec}^\phi$ for an MAEC $\mec\in \maec$ is polynomial in the size of $\mec$.

As was proved in Sec.~\ref{subsec:stratproduct}, the presented solution to Problem~\ref{pf:acpc} is correct and complete. However, the resulting optimal strategy $\strat_{\product}$ for $\product$, and hence the projected strategy $\strat$ for $\mdp$ as well, is not a finite-memory strategy in general. 
The reason is that in the second phase of every round $i$, the strategy $\strat_{\mec}^{V}$ is applied for $l_i$ surveillance cycles and $l_i$ is generally growing with $i$. 

This, however, does not prevent the solution to be effectively used. The following simple rule can be applied to avoid performing all $l_i\geq \max\{i\cdot k_i\cdot g_{\product max},j(\tfrac{1}{i})\}$ surveillance cycles in every round $i$. When the computation is in the second phase of round $i$ and the product is in an MAEC $\mec\in \maec$, after completion of every surveillance cycle, we can check whether the average cumulative cost per surveillance cycle in round $i$ is at most $V_{\mec}^*+\tfrac{2}{i}$. If yes, we can proceed to the next round $i+1$, otherwise continue with the second phase of round $i$. As the simulation results in Sec.~\ref{sec:casestudy} show, the use of this simple rule dramatically decreases the number of performed surveillance cycles in almost every round.


On the other hand, the complexity of the resulting strategy $\strat$ for $\mdp$ can be reduced from non-finite-memory to finite-memory in the following case. Assume that for every $\mec\in \maec$, the optimal ACPC strategy $\strat^V_{\mec}$ leads to an EC that contains a state from $G$, where $\mec$ is accepting with respect to the pair $(B,G)\in Acc$. In this case, the optimal strategy $\strat_{\product}$ can be defined as a finite-memory strategy that first applies the strategy $\strat_0$ to reach a state of an MAEC $\mec\in \maec$, and from that point on, only applies the strategy $\strat^V_{\mec}$. 




\section{Case Study}\label{sec:casestudy}

We implemented the solution presented in Sec.~\ref{sec:solution} in Java and applied it to a persistent surveillance robotics example~\cite{applet}. In this section, we report on the simulation results.

Consider a mobile robot moving in a partitioned environment. The motion of the robot is modeled by the initialized MDP $\mdp$ shown in Fig.~\ref{fig:casestudy}a. The set $\ap$ of atomic propositions contains two propositions $\mathtt{base}$ and $\mathtt{job}$. As depicted in Fig.~\ref{fig:casestudy}a, state 0 is the base location and state 8 is the job location. At the job location, the robot performs some work, and at the base, it reports on its job activity.

The robot's mission is to visit both base and job location infinitely many times. In addition, at least one job must be performed after every visit of the base, before the base is visited again. 
The corresponding LTL formula is
$$\phi= \G\F\,\mathtt{base} \, \wedge \, \G\F\, \mathtt{job} \, \wedge \, \G\big(\mathtt{base} \Rightarrow \X(\neg \mathtt{base}\, \U\, \mathtt{job})\big).$$
While satisfying the formula, we want to minimize the expected average cost between two consecutive jobs, \ie the surveillance proposition $\pisur = \mathtt{job}$.

In the simulation, we use a Rabin automaton $\ra_{\phi}$ for the formula that has 5 states and the accepting condition contains 1 pair. The product $\product$ of the MDP $\mdp$ and $\ra_{\phi}$ has 50 states and one MAEC $\mec$ of 19 states. The optimal set of MAECs $\maec=\{\mec\}$. 
The optimal ACPC value $V_{\mec}^*=40.5$. In Fig.~\ref{fig:casestudy}b, we list the projections of strategies $\strat_0,\strat_{\mec}^\phi$, $\strat_{\mec}^V$ for $\product$ to strategies $\strat_{init},\strat_{p1},\strat_{p2}$ for $\mdp$, respectively. The optimal strategy $\strat$ for $\mdp$ is then defined as follows. Starting from the initial state 0, apply strategy $\strat_{init}$ until a state is reached, where $\strat_{init}$ is no longer defined. Start round number 1. In $i$-th round, proceed as follows. In the first phase of the round, apply strategy $\strat_{p1}$ until the base is reached and then for one more step (the product $\product$ has to reach a state from the Rabin pair). Let $k_i$ denote the number of steps in the first phase of round $i$. In the second phase, use strategy $\strat_{p2}$ for $l_i=\max\{i\cdot k_i\cdot 10,j(\tfrac{1}{i})\}$ surveillance cycles, \ie until the number of jobs performed by the robot is $l_i$. We also use the rule described in Sec.~\ref{subsec:compldisc} to shorten the second phase, if possible. 

Let us summarize the statistical results we obtained for 5 executions of the strategy $\strat$ for $\mdp$, each of 100 rounds. The number $k_i$ of steps in the first phase of a round $i>1$ was always 5 because in such case, the first phase starts at the job location and the strategy $\strat_{p1}$ needs to be applied for exactly 4 steps to reach the base. Therefore, in every round $i>1$, the number $l_i$ is at least $50\cdot i$, \eg in round 100, $l_i\geq 5000$. However, using the rule described in Sec.~\ref{subsec:compldisc}, the average number of jobs per round was 130 and the median was only 14. In particular, the number was not increasing with the round. On the contrary, it appears to be independent from the history of the execution. In addition, at most 2 rounds in each of the executions finished only at the point, when the number of 
jobs performed by the robot in the second phase reached $l_i$. The average ACPC value attained after 100 rounds was 40.56.

In contrast to our solution, the algorithm proposed in~\cite{dennisCDC11} does not find an optimal strategy for $\mdp$. Regardless of the initialization of the algorithm, it always results in a sub-optimal strategy, namely the strategy $\strat_{p1}$ from Fig.~\ref{fig:casestudy}b that has ACPC value 50.5. 


\section{Conclusion}

In this paper, we focus on the problem of designing a control strategy for an MDP to guarantee satisfaction of an LTL formula with surveillance task, and at the same time, to minimize the expected average cumulative cost between visits of surveillance states. This problem was previously addressed in~\cite{dennisCDC11}, where the authors propose a sub-optimal solution based on dynamic programming. In contrast to this work, we exploit recent results from theoretical computer science, namely game theory and probabilistic model checking, to provide a sound and complete solution to this control problem.


\bibliographystyle{abbrv}
\bibliography{mdp_acps_acpc}

\end{document}